\documentclass[runningheads]{llncs}

\usepackage{xspace}
\usepackage{proof}
\usepackage{graphicx}
\usepackage{pdfpages}
\usepackage{amsmath}
\usepackage{amsthm}
\usepackage{amscd}
\usepackage{amssymb}
\usepackage{float}
\usepackage{epsfig}
\usepackage{xcolor}
\usepackage{color}
\usepackage{subcaption}
\usepackage{filecontents}
\usepackage{moresize}
\usepackage{stmaryrd}
\usepackage{mathrsfs}
\usepackage{mathtools}

\usepackage{xr}
\usepackage{threeparttable}
\usepackage{wrapfig}
\usepackage{epstopdf}

\usepackage[strings]{underscore}

\newcommand{\name}{{\sc Retreet}}

\newcommand{\blocks}{{\sf Blocks}}
\newcommand{\allblocks}{{\sf AllBlocks}}
\newcommand{\allconds}{{\sf AllConds}}
\newcommand{\Cond}{\textit{C}}
\newcommand{\params}{{\sf Params}}
\newcommand{\allparams}{{\sf AllParams}}
\newcommand{\allfuncs}{{\sf AllFuncs}}
\newcommand{\allcalls}{{\sf AllCalls}}
\newcommand{\allnoncalls}{{\sf AllNonCalls}}

\newcommand{\hide}[1]{} 

\newcommand{\etal}{\textit{et al.}\@\xspace}

\numberwithin{equation}{section}

\newcommand{\sif}[3]{{\sf if}~{\sf (}#1{\sf )}~#2~{\sf else}~#3}

\newcommand{\strue}{{\sf true}}

\newcommand{\snil}{{\sf nil}}

\definecolor{dkgreen}{rgb}{0,0.3,0}
\definecolor{gray}{rgb}{0.5,0.5,0.5}
\definecolor{mauve}{rgb}{0.58,0,0.82}
\definecolor{light-gray}{gray}{0.80}

\usepackage{listings}
\lstdefinelanguage{spmd}{
  morekeywords = {
       loc, bit, bool, true, false
     , implements, harness
     , null
     , assert, assume
     , else
     , find, fix, fold, for, forall, function
     , generator, gen
     , if, while, int, float, bool, string
     , loop, simple, cond, val
     , fork, join
     , nil, null, none, new, malloc
     , option, or
     , ref, return
     , spmdfork, nprocs, spmdtransfer
     , void
     , concrete, sym
     , requires, ensures
     , invariant, decreases 
     , conj, exp
     , init, stmt},
  literate=
    {-}{--}1,
  morecomment=[l]{//}
}
\renewcommand{\scriptsize}{\fontsize{8.5}{9}\selectfont}
\makeatletter
\lst@AddToHook{TextStyle}{\let\lst@basicstyle\scriptsize\fontfamily{lmss}\selectfont}
\makeatother

\usepackage{url}

\usepackage{rotating}

\usepackage{enumitem}

\usepackage{algorithm2e}
\usepackage{pstricks}
\usepackage{multirow}

\newcommand{\code}[1]{\textsf{#1}}

\newcommand{\conf}[1]{}

\newcommand{\Mona}{\textsc{Mona}\xspace}

\newcommand{\nil}{{\textit{nil}}}

\newcommand{\dir}{\textit{dir}}

\newcommand{\LExpr}{\textit{LExpr}}
\newcommand{\AExpr}{\textit{AExpr}}
\newcommand{\BExpr}{\textit{BExpr}}
\newcommand{\Stmt}{\textit{Stmt}}

\usepackage{booktabs}   

\begin{document}

\title{Reasoning About Recursive Tree Traversals}         

\author{Yanjun Wang\inst{1} \and
Jinwei Liu\inst{2} \and
Dalin Zhang\inst{2} \and
Xiaokang Qiu\inst{1} }
\institute{Purdue University\\
\email{\{wang3204,xkqiu\}@purdue.edu}
\and 
Beijing Jiaotong University\\
\email{\{12251187,dalin\}@bjtu.edu.cn}}

\maketitle

\begin{abstract}
Traversals are commonly seen in tree data structures, and performance-enhancing transformations between tree traversals are critical for many applications. Existing approaches to reasoning about tree traversals and their transformations are ad hoc, with various limitations on the class of traversals they can handle, the granularity of dependence analysis, and the types of possible transformations. We propose \name{}, a framework in which one can describe general recursive tree traversals, precisely represent iterations, schedules and dependences, and automatically check data-race-freeness and transformation correctness. The crux of the framework is a stack-based representation for iterations and an encoding to Monadic Second-Order (MSO) logic over trees. Experiments show that our framework can reason about traversals with sophisticated mutual recursion on real-world data structures such as CSS and cycletrees.
\end{abstract}

\section{Introduction}

Trees are one of the most widely used data structures in computer programming and data representations. Traversal is a common means of manipulating tree data structures for various systems, as diverse as syntax trees for compilers~\cite{petrashko17}, DOM trees for web browsers~\cite{Meyerovich2013}, and $k$-d trees for scientific simulation~\cite{rajbhandari2016fusing,rajbhandari2016sc,jo11oopsla,jo12oopsla}. Due to dependence and locality reasons, these traversals may iterate over the tree in many different orders: pre-order, post-order, in-order or more complicated, and parallel for disjoint regions of the tree. A tree traversal can be regarded as a sequence of \emph{iterations} (of each executing a code block on a tree node) \footnote{We call it an iteration because it is equivalent to a loop iteration in a loop.} and many transformations essentially tweak the order of iterations for better performance or code quality, with the hope that no dependence is violated.

Matching this wide variety of applications, orders, and transformations, there has been a fragmentation of mechanisms that represent and analyze tree traversal programs, each making different assumptions and tackling a different class of traversals and transformations, using a different formalism. For example, Meyerovich~\etal~\cite{Meyerovich2010,Meyerovich2013} use attribute grammars to represent webpage rendering passes and automatically compose/parallelize them, but the traversals representable and fusible are limited, as the dependence analysis is coarse-grained at the attribute level. TreeFuser~\cite{Sakka2017} uses a general imperative language to represent traversals, but the dependence graph it can build is similarly coarse-grained. In contrast, the recently developed PolyRec~\cite{polyrec} framework supports precise instance-wise analysis for tree traversals, but the underlying transducer representation limits the traversals they can handle to a class called perfectly nested recursion. All these mechanisms are ad hoc and incompatible, making it impossible to represent more complicated traversals or combine heterogeneous transformations. For instance, a simple, mutually recursive tree traversal is already beyond the scope of all existing approaches.

\begin{figure}[t!b]
\begin{center}
\includegraphics[scale=0.25]{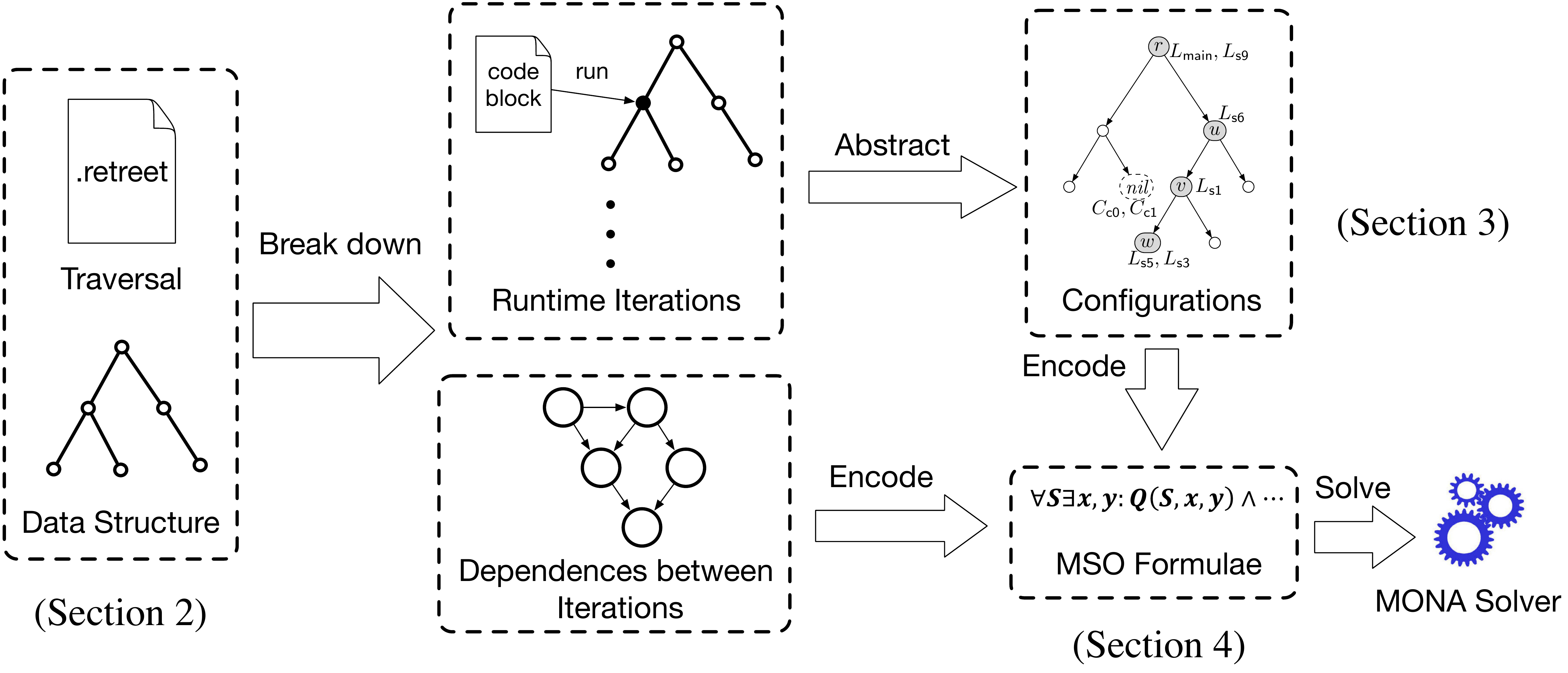}
\end{center}
\caption{\name{} Reasoning Framework}\label{fig:overview}
\end{figure}

To this end, we present \name{}, a general framework (as illustrated in Figure~\ref{fig:overview}) in which one can write almost arbitrary tree traversals, reason about dependences between iterations of fine granularity, and check correctness of transformations automatically. This framework features an abstract yet detailed characterization of iterations, schedules and dependences, which we call \emph{Configuration}, as well as a powerful reasoning algorithm. In this paper, we first present \name{} (``REcursive TREE Traversal'') as an expressive language that allows the user to flexibly describe tree traversals in a recursive fashion (Section~\ref{sec:syntax}). Second, we propose Configuration as a detailed, stack-based abstraction for dynamic instances in a traversal (Section~\ref{sec:configuration}). This abstraction can be encoded to Monadic Second-Order logic over trees, which allows us to reason about dependences and check data-race-freeness and equivalence of \name{} programs (Section~\ref{sec:encoding}). Finally, we show this framework is practically useful by checking the correctness of four different classes of tree traversals, including fusing and parallelizing real-world applications such as CSS minification and Cycletree routing (Section~\ref{sec:case}).

\section{A Tree Traversal Language}
\label{sec:syntax}

\begin{figure}[t!b]
\scriptsize
\begin{displaymath}
\begin{array}{lllll}
\dir \in \code{Loc}~\textrm{Fields}~~ & v \in \code{Int}~\textrm{Vars}~~ & n \in \code{Loc}~\textrm{Vars}~~ & f \in \code{Int}~\textrm{Fields}~~ & g : \textrm{Function} \textrm{~IDs}
\end{array}
\end{displaymath}
\begin{displaymath}
\begin{array}{rcl}
\LExpr & ::= & n ~\big|~ \LExpr.\dir \\
\AExpr & ~::=~ & 0 ~\big|~ 1 ~\big|~ n.f ~\big|~ v ~\big|~ \AExpr + \AExpr ~\big|~ \AExpr -\AExpr \\
\BExpr & ::= & LExpr == \snil ~\big|~ \strue ~\big|~ \AExpr > 0 ~\big|~
 \code{!} ~\BExpr ~\big|~ \BExpr~ \code{\&\&} ~\BExpr \\
\textit{Assgn} & ::= & n.f = \AExpr ~\big|~ v = \AExpr ~\big|~ \code{return}~\bar{v}\\ 
\textit{Block} & ::= & \bar{v} = t(\LExpr, \overline{\AExpr}) ~\big|~ \textit{Assgn}^+ \\
\Stmt & ::= & 
 \textit{Block} ~\big|~ \sif{\BExpr}{\Stmt}{\Stmt} ~\big|~ \Stmt~;~\Stmt ~\big|~ \{ \Stmt \parallel \Stmt \} \\
\textit{Func} & ::= & g(n, \bar{v}) \{~ \Stmt ~\}~\footnote{sss sdss} \\
\textit{Prog} & ::= & \textit{Func}^+  \\
\end{array}
\end{displaymath}
$^1$ Any function $g(n, \bar{v})$ should not contain recursive calls to $g(n, \dots)$, regardless of directly in $\Stmt$ or indirectly through inlining arbitrarily many calls in $\Stmt$.
\caption{Syntax of \name{}}\label{fig:syntax}
\end{figure}

In this section, we present \name{}, our imperative, general tree traversal language. \name{} programs execute on a tree-shaped heap which consists of a set of locations. Each location, also called node, is the root of a (sub)tree and associated with a set of pointer fields $dir$ and a set of local fields $f$. Pointer fields $dir$ contains the references to the children of the original location; local fields stores the local ${\sf Int}$ values. 

The syntax of \name{} is shown in Figure~\ref{fig:syntax}. A program consists of a set of functions; each has a single \code{Loc} parameter and optionally, a vector of \code{Int} parameters. We assume every program has a \code{Main} function as the entry point of the program. 
The body of a function comprises \textit{Blocks} of code combined using conditionals, sequentials and parallelizations.

A block of code is either a function call or a straight-line sequence of assignments. A function call takes as input a \LExpr~which can be the current \code{Loc} parameter or any of its descendant, and a sequence of \AExpr's of length as expected. Each \AExpr~is an integer expression combining \code{Int} parameters and local fields of the \code{Loc} parameter. Non-call assignments compute values of \AExpr's and assign them to \code{Int} parameters, fields or special return variables. Note that the functions in \name{} can be mutually recursive, i.e., two or more functions call each other. However, there is a special syntactic restriction: every function $g(n, \bar{v})$ should not call, directly or indirectly through inlining, itself, i.e., $g(n, \dots)$ with arbitrary \code{Int} arguments (see more discussion below).

The semantics of \name{} is common as expected and we omit the formal definition. In particular, all function parameters are call-by-value; the parallel execution adopts the statement-level interleaving semantics (every execution is a serialized interleaving of atomic statements).

\begin{figure}[t!]
\scriptsize
\begin{subfigure}[c]{0.32\columnwidth}
\begin{lstlisting}[basicstyle=\sffamily \scriptsize]
  Odd(n)
    if (n == nil)	// c0
    	return 0	// s0
    else
    	ls = Even(n.l)			// s1
    	rs = Even(n.r)			// s2
    	return ls + rs + 1	// s3
\end{lstlisting}
\end{subfigure}
\begin{subfigure}[c]{0.33\columnwidth}
\begin{lstlisting}[basicstyle=\sffamily \scriptsize]
  Even(n)
    if (n == nil)		// c1
    	return 0		// s4
    else
    	ls = Odd(n.l)		// s5
    	rs = Odd(n.r)		// s6
    	return ls + rs	// s7
\end{lstlisting}
\end{subfigure}
\begin{subfigure}[c]{0.32\columnwidth}
\begin{lstlisting}[basicstyle=\sffamily \scriptsize]
  Main(n)
    { 
      o = Odd(n) $\parallel$	// s8
      e = Even(n)			// s9
    }
    return (o, e)			// s10
\end{lstlisting}
\end{subfigure}
\caption{Example of mutually recursive tree traversals}
\label{fig:number}
\end{figure}

\paragraph{Example:} Figure~\ref{fig:number} illustrates our running example, which is a pair of mutually recursive tree traversals. \code{Odd(n)} and \code{Even(n)} count the number of nodes at the odd and even layers of the tree \code{n}, respectively (\code{n} is at layer 1, \code{n.l} is at layer 2, and so forth). \code{Odd} and \code{Even} recursively call each other; and the \code{Main} function runs \code{Odd} and \code{Even} in parallel, and return the two computed numbers.
Note that the mutual recursion is beyond the capability of existing automatic frameworks that handle tree traversals~\cite{Amiranoff,Weijiang2015,Meyerovich2010,Meyerovich2013,Sakka2017,polyrec}.


\subsection{Discussion of the Language Design}

We remark about some critical design features of \name{}. In a nutshell, \name{} has been carefully designed to be \emph{maximally permissible} of describing tree traversals, yet \emph{encodable} to the MSO logic. More specifically, three major design features make possible our MSO encoding presented in Section~\ref{sec:encoding}: \emph{obviously terminating}, \emph{single node traversal} and \emph{no-tree-mutation}. Despite these restrictions, \name{} is still more general and more expressive than the state of the art---to the best of our knowledge, all the restrictions we discuss below can be seen in all existing approaches (find more discussion in Section~\ref{sec:related}).

{\bf Termination:} \name{} describes \emph{obviously terminating} tree traversals. Note that the ``$g(n, \bar{v})$ does not call any $g(n, \dots)$'' restriction does not only guarantee the termination, but also bounds the steps of executions. With this restriction, every function call makes progress toward traversing the tree downward. Hence, the height of the call stack will be bounded by the height of the tree, and every statement~\footnote{Notice that two different call sites of the same function are considered two different statements. So the number of statements is bounded by the size of the program.} is executed on a node at most once. Therefore, running a \name{} program $P$ on a tree $T$ will terminate in $\mathcal{O}(|P|^{h(T)})$ steps where $h(T)$ is the height of the tree. This bound is critical as it allows us to encode the program execution to a tree model, with only a fixed amount of information on each node. In contrast, \name{} excludes the following program: \code{A(n, k): if (k $<$= 0) return 0; else return A(n, k-1) + ...}
The program terminates, but the length of execution on node \code{n} is determined by the input value \code{k}, which can be arbitrarily large and makes our tree-based encoding impossible. 

{\bf Single node traversal:} In \name{}, all functions take only one \code{Loc} parameter. Intuitively, this means the tree traversal is not allowed to manipulate more than one node at one time. This is a nontrivial restriction and necessary for our MSO encoding. The insight of this restriction will be clearer in Section~\ref{sec:encoding}.

{\bf No tree mutation:} Mutation to the tree topology is generally disallowed in \name{}. General tree mutations will possibly affect the tree-ness of the topology, where our tree-based encoding can not fit in. However, we can simulate a limited class of tree mutations by using mutable local fields. See more details in our tree-mutation example in Section~\ref{sec:case}.

For the simplification of presentation, \name{} focuses on programs without loops or global variables. These restrictions are not essential because loops or global variables can be rewritten to recursion and local variables, respectively. As long as the rewritten program satisfies the real restrictions we set forth above, it can be handled by our framework. See our discussions below.

\paragraph{Loop-freeness:} \name{} does not allow iterative loops. Recall that \name{} is meant to describe tree traversals, and the no-self-call syntactic restriction guarantees that the program manipulates every node only a bounded number of times, and hence the termination of the program. Similarly, a typical loop or even nested loop traversing a tree only computes a limited number of steps on each node, and can be naturally converted to recursive functions in \name{}.

\paragraph{No global variables:} We omit global variables in \name{}. However, it is not difficult to extend for global variables. Note that when the program is sequential, i.e., no concurrency, one can simply replace a global variable with an extra parameter for every function, which copies in and copies out the value of the global variable. In the presence of concurrency, we need to refine the current syntax to reason about the schedule of manipulations to global variables. Basically, every statement accessing a global variable forms a separate \textit{Block}, so that we can compare the order between any two global variable operations.

\vspace{.1in}
In the rest of the paper, we also assume: all trees are binary with two pointer fields \code{l} and \code{r}, every function only calls itself or other functions on \code{n.l} or \code{n.r}, and returns only a single {\sf Int} value, and every boolean expression is atomic, i.e., of the form $LExpr == \snil$ or $\AExpr > 0$. In addition, we assume the program is free of null dereference, i.e., every term $\textit{le}.\dir$ is preceded by a guard $\textit{le} \code{ != nil}$. Note that relaxing these assumptions will not affect any result of this paper, because any \name{} program violating these assumptions can be easily rewritten to a version satisfying the assumptions.

\section{Iteration Representation}
\label{sec:configuration}

We consider code blocks as atomic units of \name{} programs. Code blocks (function calls or straight-line assignments) are building blocks of \name{} programs and are a key to our framework. In our running example (Figure~\ref{fig:number}), there are 11 blocks. We number the blocks with \code{s0} through \code{s10}, as shown in the comment following each block. Then the execution of a \name{} program is a sequence of iterations, each running a non-call code block on a tree node. For example, consider executing our running example on a single-node $u$ (i.e., $u.{\sf l} = u.{\sf r} = {\sf nil}$), one possible execution is a sequence of iterations (also called instances in the literature): $({\sf s0}, u.l), ({\sf s0}, u.r), ({\sf s7}, u), ({\sf s4}, u.l), ({\sf s4}, u.r), ({\sf s3}, u)$. Note that every iteration is unique and appears at most once in a traversal, as per the syntactic restriction of \name{}. However, this representation is not sufficient to reason about the dependences between steps. For example, if the middle steps $({\sf s7}, u), ({\sf s4}, u.l)$ were swapped, is that still a possible sequence of execution? The question can't be answered unless we track back the contexts in which the two steps are executed: $({\sf s7}, u)$ is executed in the call to \code{Even($u$)} (block \code{s9}); $({\sf s4}, u.l)$ is executed in the call to \code{Even($u$.l)}, which is further in \code{Odd($u$)} (block \code{s8}). As the two calls are running in parallel, swapping the two steps yields another legal sequence of execution. Automating this kind of reasoning is extremely challenging. In fact, even determining if an instantiation exists is already undecidable:
\begin{theorem}\label{thm:undecidability}
Determining if an iteration may occur in a \name{} program execution is undecidable.
\end{theorem}
\begin{proof}
See Appendix~\ref{sec:proof} of~\cite{extended}.
\end{proof}

\subsection{Configurations}

\begin{figure}[t!b]
\scriptsize
\begin{subfigure}[b]{.45\textwidth}
\centering
\begin{tabular}{|c|c|}
\hline
Record Num & Content \\
\hline
~0~ & $~(\code{main}, r, \code{s8}=5, \dots)~$ \\
\hline
~1~ & $~(\code{s9}, r, \code{s5}=3, \dots)~$ \\
\hline
~2~ & $~(\code{s6}, u, \dots)~$ \\
\hline
~3~ & $~(\code{s1}, v, \dots)~$ \\
\hline
~4~ & $~(\code{s5}, w, \code{s1}=0, \code{s2}=0)~$ \\
\hline
~5~ & $~(\code{s3}, w)~$ \\
\hline
\end{tabular}
\caption{A configuration}
\label{fig:stack}
\end{subfigure}
\hfill
\begin{subfigure}[b]{.5\textwidth}
\centering
\includegraphics[scale = 0.8]{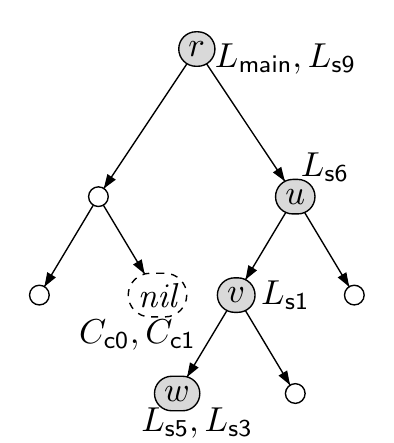}
\caption{Represented as labels on the tree}
\label{fig:labeling}
\end{subfigure}
\caption{Example of Configuration Encoding}
\label{fig:example}
\end{figure}

As precise reasoning about \name{} is undecidable, we propose an iteration representation called \emph{configuration}, which is a right level of abstraction for which automated reasoning is possible. Intuitively, a configuration looks like a snapshot of the call stack. The top record describes the current running block as we discussed above. Each other record describes a call context which includes: the callee block, the single \code{Loc} parameter, and other \code{Int} variables' values. These \code{Int} values are a bit unusual: first, for each \code{Int} parameter, the context records its \emph{initial} value received when the call begins; second, for each function call within the current call, the context uses \emph{ghost variables} to predict the return values.

\paragraph{Example:} Figure~\ref{fig:stack} gives an example of a configuration, which consists of 6 records. The top record indicates that the current step is running block \code{s3} on tree node $w$, and the current values of local variables. In other records, we only show the callee stack, the \code{Loc} parameter, and other relevant \code{Int} variables. For example, the value ${\sf s8}=5$ means that the call in \code{s8} is predicted to finish and return value $5$, which might be relevant to the next call context, \code{s9}.

\vspace{.1in}
Obviously, all stacks of records are not valid configurations. In particular, the beginning record should run \code{main} and the last record should run a non-call block. More importantly, for any non-beginning record, 
the path condition of the block should be satisfied, i.e., this block of code can be reached from the beginning of the function it belongs to.
While a precise characterization of these constraints is expensive and leads to undecidability as per Theorem~\ref{thm:undecidability}, our key idea is to loosen these constraints using an abstraction called \emph{speculative execution}:

\begin{definition}[Speculative Execution]
Given a function $f$, a group of initial values $I: \params(f) \rightarrow \mathbb{Z}$ and a group of speculative outputs $O: \blocks(f) \cap \allcalls \rightarrow \mathbb{Z}$, a speculative execution of $f$ with respect to $I$ and $O$ follows the following steps:
1) initialize every parameter $p$ with value $I(p)$, and let the current block ${\sf c}$ be the first block in $f$;
2) if ${\sf c}$ is not a call, then simulate the execution of ${\sf c}$, and move to the next block;
3) if ${\sf c}$ is a call of the form $v = g(le, \bar{ie})$, then update $v$'s value with $O({\sf c})$.
\end{definition}

Intuitively, speculative execution abstracts normal execution of a recursive function by replacing all recursive calls with a speculative return value, which is given as an input at the beginning of the execution. With the speculative execution we can now formally define configuration, which overapproximates real configurations possible in an execution.

\begin{definition}[Configuration]
\label{def:configuration}
A configuration of length $k$ on a tree $T$ is a mapping $\mathcal{C} : [k] \rightarrow \allblocks \times {\sf Nodes}(T) \times (\allparams \cup \allcalls \rightharpoonup \mathbb{Z})$ such that:
\begin{itemize}
\item For any $0 \leq i < k$, $\mathcal{C}(i)$ is of the form $({\sf s}, u, M)$ where ${\sf s} \in \allcalls$ is a call to a function $f$, and $M$ is only defined on $\params(f) \cup \blocks(f)$.
\item The last record $\mathcal{C}(k)$ is of the form $({\sf s}, u, \emptyset)$, where ${\sf s} \in \allnoncalls$. 
\item The first record $\mathcal{C}(0)$ is of the form $({\sf main}, \textit{root}_T, ...)$.
\item For any two adjacent records $\mathcal{C}(i-1) = ({\sf s}, u, M)$, $\mathcal{C}(i) = ({\sf t}, v, N)$, 
${\sf s}$ is a call to the function that ${\sf t}$ belongs to (denoted as ${\sf s} \triangleleft {\sf t}$).
Moreover, speculatively executing the function with respect to initial values $M|_{\params(g)}$ and speculative outputs $M|_{\blocks(g)}$ leads to record $({\sf t}, v, N)$. 
\end{itemize}
\end{definition}
\noindent Note that all the sets mentioned in Definition~\ref{def:configuration} (e.g., $\allblocks, \allparams, \allcalls$) are self explanatory and we leave their definitions in Figure~\ref{fig:sets} of Appendix~\ref{sec:block}  of~\cite{extended}. We call the last part of the definition above \emph{reachability} from $({\sf s}, u, M)$ to $({\sf t}, v, N)$, and show that the reachability can be represented as a logical constraint:
\begin{lemma}\label{thm:pathcond}
Let $({\sf s}, u, M)$ and $({\sf t}, v, N)$ be two records such that ${\sf s} \triangleleft {\sf t}$.
Then $({\sf s}, u, M)$ reaches $({\sf t}, v, N)$ if $(u, v, M, N)$ satisfies 
\[\textit{PathCond}_{\code{s,t}}(u, v, M, N) ~~\equiv~~ \displaystyle \textit{Match}_{\code{s,t}}(u, v, M, N) \wedge \bigwedge_{\code{c} \in \code{Path(t)}} \textit{WP}(\code{c}, M)  \] 
\end{lemma}
\noindent Here $\textit{Match}_{\code{s,t}}(u, v, M, N)$ says speculative execution starting from $M$ leads to $N$; $\textit{WP}(\code{c}, M)$ says the weakest precondition of \code{c} is satisfied by $M$ (details in Appendix~\ref{sec:wp} of~\cite{extended}).

\paragraph{Example:} 
Consider a code block \code{s} calling a function \code{func(n, p, r0) \{ n.f = p + 1 ; r1 = r0; if (n.f $<$ r1) \{...\} else \{ t \} \}}, where \code{t} is a recursive call to \code{func(n.l)}. For record $({\sf s}, u, M)$ to reach record $({\sf t}, v, N)$, there is only one condition, \code{n.f $<$ r1}, which occurs negatively. In other words, 
the code sequence reaching \code{t} is \code{ n.f = p + 1 ; r1 = r0; assume (n.f $\geq$ r1); t }. 
In addition, since code block $s$ and $t$ invoke function \code{func} on node $n$ and $n.l$ respectively, $\textit{Match}(u, v, M, N)$ should ensure that $v$ is the left child of $u$, i.e. in this case, $\textit{Match}(u, v, M, N) \equiv u.l = v$.
Therefore the path condition can be computed as $\textit{PathCond}_{\code{s,t}}(u, v, M, N) \equiv M(\code{p})+1 \geq M(\code{r0}) \wedge u.l = v$. 


\section{Encoding to Monadic Second-Order Logic}
\label{sec:encoding}

In this section, we show that fine-grained dependence analysis problems for \name{} can be encoded to Monadic Second-Order (MSO) logic over trees, a well known decidable logic. The syntax of the logic contains a unique $\textit{root}$, two basic operators $\textit{left}$ and $\textit{right}$. There is a binary predicate $\textit{reach}$ as the transitive closure of $\textit{left}$ and $\textit{right}$, and a special \textit{isNil} predicate with constraint $\forall v. \big( \textit{isNil}(v) \rightarrow \textit{isNil}(\textit{left}(v)) \wedge \textit{isNil}(\textit{right}(v)) \big)$. 

\subsubsection*{Encoding Configurations.}

First of all, we need to encode configurations we presented in Section~\ref{sec:configuration}.
Given a \name{} program, we define the following labels (each of which is a second-order variable):
\begin{itemize}
\item for each code block \code{s}, introduce a label (a second-order variable) $L_{\code{s}}$ such that $L_{\code{s}}(u)$ denotes that there exists a record $(\code{s}, u, \dots)$ in the configuration; 
\item for each branch condition \code{c}, introduce a label $\Cond_{\sf c}$ such that $\Cond_{\sf c}(u)$ denotes that $\textit{WP}(\code{c}, M)$ is satisfied by a record of the form $(\code{s}, u, M)$;
\item for each pair of blocks ${\sf s}$ and ${\sf t}$ such that ${\sf s} \triangleleft {\sf t}$, introduce a label $K_{\code{s}, \code{t}}$ such that $K_{\code{s}, \code{t}}(u,v)$ denotes that $\textit{Match}_{\code{s,t}}(u, v, M, N)$ is satisfied by records $({\sf s}, u, M)$ and $({\sf t}, v, N)$.
\end{itemize}
\noindent Note that these labels allow us to build an MSO predicate $\overline{\textit{PathCond}_{\code{s,t}}}$ as an abstracted version of the path condition $\textit{PathCond}_{\code{s,t}}$ defined in Lemma~\ref{thm:pathcond}:
\[\overline{\textit{PathCond}_{\code{s,t}}}(u, v) ~~\equiv~~ \displaystyle K_{\code{s,t}}(u, v) \wedge \bigwedge_{\code{c} \in \code{Path(t)}} \Cond_{\sf c}(u)  \] 

\paragraph{Example:} The configuration in Figure~\ref{fig:stack} can be encoded to labels on the tree in Figure~\ref{fig:labeling}. Note that the labels $\Cond_{\sf c0}$ and $\Cond_{\sf c1}$ are labeled on $\nil$ nodes only. If a node has a particular label, the node belongs to the set represented by the corresponding second-order variable. For example, node $u$ is in $L_{s6}$ but nodes $r, v$ and $w$ are not.

\vspace{.1in}
As the set of blocks and the set of conditions are fixed and known, we can simply represent these second-order variables using labeling predicates $L \subseteq \allcalls \cup \allnoncalls \times {\sf Nodes}(T)$ and $\textit{Cond} \subseteq \allconds \times {\sf Nodes}(T)$ such that $L(\code{s}, u)$ if and only if $L_{\code{s}}(u)$, $\Cond({\code{c}}, u)$ if and only if $\Cond_{\code{c}}(u)$.~\footnote{The sets mentioned here are self explanatory and their definitions can be found in Figure~\ref{fig:sets} in Appendix~\ref{sec:block} of~\cite{extended}.} In other words, $L(\code{s}, u)$ is the syntactic sugar for $L_{\code{s}}(u)$ and $\Cond({\code{c}}, u)$ is the syntactic sugar for $\Cond_{\code{c}}(u)$.


Now we are ready to encode configurations to MSO. Below we define a formula $\textit{Configuration}(L, \Cond, \code{q}, v)$ which means $L$ and $\Cond$ correctly represent a configuration with $(\code{q}, v, \dots)$ as the current record, for some non-call block $\code{q}$:
\begin{displaymath}
\begin{array}{l}
\textit{Configuration}(L, \Cond, {\sf q}, v) \equiv  L(\code{main}, \textit{root}) \\
\qquad \land~ \textit{Current}(L, \code{q}, v) 
 \land~ \forall u. \big( u \neq v \rightarrow \displaystyle \bigwedge\limits_{\code{s} \in \allnoncalls} \neg L(\code{s}, u) \big) \\ 
\qquad \land~ \forall u. \displaystyle \bigwedge\limits_{\code{s} \in \allcalls } \Big( L(\code{s}, u) \rightarrow \bigvee\limits_{\code{s} \triangleleft \code{t}} \big( \textit{Next}(L, \Cond, u, \code{s}, \code{t}) \wedge \bigwedge\limits_{\code{t} \sim \code{t'}, \code{t} \neq \code{t'}} \neg \textit{Next}(L, \Cond, u, \code{s}, \code{t'}) \big) \Big) \\
\qquad \land~ \forall u. \displaystyle \bigwedge\limits_{\code{t} \in \allcalls \cup \allnoncalls } \Big( L(\code{t}, u) \rightarrow \textit{Prev}(L, \Cond, u, \code{t})  \Big) \\
\qquad \land~ \forall u. \displaystyle \bigvee\limits_{{\sf C} \in {\sf ConsistentCondSet}} \Big( \bigwedge\limits_{{\sf c} \in {\sf C}} \Cond({\sf c}, u) \wedge \bigwedge\limits_{{\sf c} \notin {\sf C}} \neg \Cond({\sf c}, u) \Big)
\end{array}
\end{displaymath}
The first two lines claim that \code{main} is marked on the \textit{root}, and \code{q} is the only non-call block marked on the tree, where $\textit{Current}(L, \code{q}, v)$ is a subformula indicating that for the current node $v$, a record $(\code{q}, v, \dots)$ is in the stack for exactly one non-call block $\code{q}$: $\textit{Current}(L, \code{q}, v) \equiv L(\code{q}, v) \land \bigwedge\limits_{\code{q'} \in \allnoncalls, \code{q'} \neq \code{q}} \neg L(\code{q'}, v) $.

The next two lines, intuitively, say that every record has a unique successor (and predecessor) that can reach to (and from). $\code{t} \sim \code{t'}$ denotes that $\code{t}$ and $\code{t'}$ are from the same function. Predicates $\textit{Next}$ and $\textit{Prev}$ are defined as below: 
$$\textit{Next}(L, \Cond, u, \code{s}, \code{t}) \equiv \exists v.  \Big( L(\code{t}, v) \wedge \overline{\textit{PathCond}_{\code{s,t}}}(u, v) \Big)
$$
\begin{displaymath}
\begin{array}{l}
\textit{Prev}(L, \Cond, u, \code{t}) \equiv \exists v. \bigg( \displaystyle \bigvee\limits_{\code{s} \triangleleft \code{t}} \Big( L(\code{s}, v) \wedge \overline{\textit{PathCond}_{\code{s,t}}}(v, u)
\\
\qquad \qquad \qquad \qquad \qquad \qquad \wedge \bigwedge\limits_{\code{s'} \triangleleft \code{t}, s' \neq s } \neg \big( L(\code{s'}, v) \wedge \overline{\textit{PathCond}_{\code{s,t}}}(v, u)
\big) \Big) \bigg) \\
\end{array}
\end{displaymath}

\vspace{.1in}
The last line makes sure that for each node $u$, the set of satisfied conditions ${\sf C}$ is consistent, i.e., $\displaystyle \bigwedge\limits_{{\sf c} \in {\sf C}} \textit{WP}(\code{c}, M)$ is satisfiable for any record $(\code{s}, u, M)$. In other words, a consistent conditional set for a node $u$ represents a feasible conditional path from the root of the tree to reach node $u$. Notice that this is a linear integer arithmetic constraint and SMT-solvable. Hence we can assume the set of all possible consistent condition set, ${\sf ConsistentCondSet}$, has been computed a priori.

%

\subsubsection*{Schedules and Dependences.}

The definition and encoding of configurations above have paved the way for reasoning about \name{} programs. Given two configurations, a basic query one would like to make is about their order in a possible execution: can the two configurations possibly coexist? If so, are they always ordered? Or they can occur in arbitrary order due to the parallelization between them? To answer these questions, intuitively, we need to pairwisely compare the records in the two configurations from the beginning and find the place that they diverge. We define the following predicate:
\begin{displaymath}
\begin{array}{l}
\textit{Consistent}_{\sf s, t_1, t_2}(L_1, L_2, \Cond_1, \Cond_2) \equiv \exists z. \Big[ \\
 \forall v. \bigg( \textit{reach}(v,z) \rightarrow \Big( \displaystyle \bigwedge\limits_{\code{s}} \big( L_1(\code{s}, v) \leftrightarrow L_2(\code{s}, v) \big) ~\land~ \displaystyle \bigwedge\limits_{\code{c}} \big( \Cond_1(\code{c}, v) \leftrightarrow \Cond_2(\code{c}, v) \big) \Big) \bigg) \\
 \land~   L_1(\code{s}, z) \land L_2(\code{s}, z) \land \textit{Next}(L_1, \Cond_1, z, \code{s}, \code{t}_1) \land \textit{Next}(L_2, \Cond_2, z, \code{s}, \code{t}_2) \Big]
\end{array}
\end{displaymath}
The predicate assumes there are two sequences of records represented as $(L_1, \Cond_1)$ and $(L_2, \Cond_2)$, respectively, and indicates that there is a diverging record $({\sf s}, z, \dots)$ in both sequences such that: 1) the two configurations match on all records prior to the diverging record; 2) the next records after the diverging one are $({\sf t_1}, \dots)$ and $({\sf t_2}, \dots)$, respectively, and they can be reached at the same time (i.e., $\Cond_1$ and $\Cond_2$ agree on the diverging node $z$). 

${\sf t_1}$ and ${\sf t_2}$ are obviously in the same function and there are two possible relations between them: a) if ${\sf t_1}$ precedes ${\sf t_2}$ (or symmetrically, ${\sf t_2}$ precedes ${\sf t_1}$), then configuration $(L_1, \Cond_1)$ always precedes $(L_2, \Cond_2)$ (or vice versa); b) otherwise, ${\sf t_1}$ and ${\sf t_2}$ must be two parallel blocks, then the two configurations occur in arbitrary order. Both the two relations can be described in MSO (see Figure~\ref{fig:relation}). 

\begin{figure}[t!]
\begin{displaymath}
\begin{array}{l}
\textit{Ordered}(L_1, L_2, \Cond_1, \Cond_2) \equiv  \displaystyle \bigvee\limits_{\code{s},\code{t}_1,\code{t}_2 \atop{\code{s} \triangleleft \code{t}_1, \code{s} \triangleleft \code{t}_2, \code{t}_1 \prec \code{t}_2 } } \textit{Consistent}_{\sf s, t_1, t_2}(L_1, L_2, \Cond_1, \Cond_2)
\end{array}
\end{displaymath}
\begin{displaymath}
\begin{array}{l}
\textit{Parallel}(L_1, L_2, \Cond_1, \Cond_2) \equiv  \displaystyle \bigvee\limits_{\code{s},\code{t}_1,\code{t}_2 \atop{\code{s} \triangleleft \code{t}_1, \code{s} \triangleleft \code{t}_2, \code{t}_1 \parallel \code{t}_2 } } \textit{Consistent}_{\sf s, t_1, t_2}(L_1, L_2, \Cond_1, \Cond_2)
\end{array}
\end{displaymath}
\caption{Relations between consistent configurations}
\label{fig:relation}
\end{figure}


Another set of relations is necessary to describe the data dependences. We use a read\&write analysis to compute the read set $R_{\sf s}$ and write set $W_{\sf s}$ for each non-call block \code{s} (details in Appendix~\ref{sec:block} of~\cite{extended}). These sets allow us to define two binary predicates: $\textit{Write}_{\code{s}}(u, v)$ if running \code{s} on $u$ will write to $v$; $\textit{ReadWrite}_{\code{s}}(u, v)$ if running \code{s} on $u$ will read or write to $v$.
The following predicate describes two configurations $(L_1, \Cond_1, {\sf s}, u)$ and $(L_2, \Cond_2, {\sf t}, v)$ with data dependence: both last records $({\sf s}, u, \dots)$ and $({\sf t}, v, \dots)$ access the same node $z$ and at least one of the accesses is a write:
\begin{displaymath}
\begin{array}{l}
\textit{Dependence}_{\sf s,t}(u, v, L_1, L_2, \Cond_1, \Cond_2) \equiv \textit{Configuration}(L_1, \Cond_1, {\sf s}, u) \wedge \textit{Configuration}(L_2, \Cond_2, {\sf t}, v) \\
\quad \land~ \exists z. \Big( \big(\textit{ReadWrite}_{\code{s}}(u, z) \wedge \textit{Write}_{\code{t}}(v, z)\big) \vee \big(\textit{Write}_{\code{s}}(u, z) \wedge \textit{ReadWrite}_{\code{t}}(v, z)\big) \Big)
\end{array}
\end{displaymath}

\subsubsection*{Data Race Detection and Equivalence Checking.}

Now we are ready to encode some common dependence analysis queries to MSO. 
A data race may occur in a \name{} program $P$ if there exist two parallel configurations between which there is data dependence:
\begin{displaymath}
\begin{array}{l}
\textit{DataRace}\llbracket P \rrbracket \equiv \displaystyle \bigvee\limits_{{\sf q_1, q_2} \in \allnoncalls} \exists x_1, x_2, L_1, L_2, \Cond_1, \Cond_2. \Big(  \\
\quad   \textit{Dependence}_{\sf q_1,q_2}(x_1, x_2, L_1, L_2, \Cond_1, \Cond_2) \land \textit{Parallel}(L_1, L_2, \Cond_1, \Cond_2) ~\Big)
\end{array}
\end{displaymath}

\begin{theorem}
A \name{} program $P$ is data-race-free if $\textit{DataRace}\llbracket P \rrbracket$ is invalid.
\end{theorem}

\noindent Besides data race detection, another critical query is the equivalence between two \name{} programs, which is common in program optimization. For example, when two sequential tree traversals \code{A(); B()} are fused into a single traversal \code{AB()}, one needs to check if this optimization is valid, i.e., if \code{A(); B()} is equivalent to \code{AB()}. Again, while the equivalence checking is a classical and extremely challenging problem, we focus on comparing programs that are built on the same set of straight-line blocks and simulate each other. The comparison is sufficient since the goal of \name{} framework is to automate the verification of common program transformations such as fusion or parallelization, which only reorder the operations of a program.

\begin{definition}
Two \name{} programs $P$ and $P'$ bisimulate if $\allnoncalls(P) = \allnoncalls(P')$ and there exists a relation $R \subseteq \allcalls(P) \times \allcalls(P')$ such that
\begin{itemize}
\item for any ${\sf s} \in \allcalls(P)$ and ${\sf s'} \in \allcalls(P')$, if ${\sf s} \triangleleft {\sf q}$ and ${\sf s'} \triangleleft {\sf q}$ for some non-call block ${\sf q}$, then $({\sf s}, {\sf s'}) \in R$.
\item if ${\sf s} \triangleleft {\sf t}$, ${\sf s'} \triangleleft {\sf t'}$ and ${(\sf t}, {\sf t'}) \in R$, then ${(\sf s}, {\sf s'}) \in R$.
\item for any node $u, v$ and any ${(\sf s}, {\sf s'}) \in R, {(\sf t}, {\sf t'}) \in R$, $\textit{PathCond}_{\code{s,t}}(u, v, M, N)$ and $\textit{PathCond}_{\code{s',t'}}(u, v, M, N)$  are equivalent.
\end{itemize}
\end{definition}
Intuitively, $P$ and $P'$ bisimulate if any configuration for $P$ can be converted to a corresponding configuration for $P'$, and vice versa.  It is not hard to develop a naive bisimulation-checking algorithm to check if two \name{} programs $P$ and $P'$ bisimulate: just enumerate all possible relations between $P$ calls and $P'$ calls, by brute force. In our experiments, we manually did the enumeration but following some automatable heuristics, e.g., giving priorities to relations that preserve the order of the calls.

The correspondence between configurations can be extended to executions, i.e., every execution of $P$ corresponds to an execution of $P'$ that runs exactly the same blocks of code on the same nodes, and vice versa. To guarantee the equivalence, it suffices to make sure that the correspondence does not swap any pair of ordered configurations with data dependences.~\footnote{We assume both programs are free of data races; otherwise the equivalence between them is undefined.} 
In the following formula, the predicates $\textit{Dependence}^P_{\sf q_1,q_2}$ and $\textit{Dependence}^{P'}_{\sf q_1,q_2}$ guarantee four configurations, two on $P$ and two on $P'$, and pair-wisely bisimulating (as they end with the same blocks).
\begin{displaymath}
\begin{array}{l}
\textit{Conflict}\llbracket P, P' \rrbracket \equiv \displaystyle \bigvee\limits_{{\sf q_1, q_2} \in \allnoncalls} \exists x_1, x_2, L_1, L_2, \Cond_1, \Cond_2, L'_1, L'_2, \Cond'_1, \Cond'_2 . \Big( \\
\quad~ \textit{Dependence}^P_{\sf q_1,q_2}(x_1, x_2, L_1, L_2, \Cond_1, \Cond_2)  ~\land~  \textit{Dependence}^{P'}_{\sf q_1,q_2}(x_1, x_2, L'_1, L'_2, \Cond'_1, \Cond'_2) \\
\quad ~\land~  \textit{Ordered}^P(L_1, L_2, \Cond_1, \Cond_2) ~\land~ \textit{Ordered}^{P'}(L'_2, L'_1, \Cond'_2, \Cond'_1) ~ \Big)
\end{array}
\end{displaymath}

\begin{theorem}
For any two data-race-free \name{} programs $P$ and $P'$ that bisimulate, they are equivalent if $\textit{Conflict}\llbracket P, P' \rrbracket$ is invalid.
\end{theorem}

\section{Evaluation}
\label{sec:case}

We have prototyped the \name{} framework \footnote{Available at: https://github.rcac.purdue.edu/wang3204/Retreet} and evaluated the effectiveness and efficiency of the framework through four case studies: a mutually recursive size-counting traversal, a tree-mutating traversal, a set of CSS minification traversals, and a cycletree traversal algorithm. Note that both of the mutually recursive size-counting traversal and the cycletree traversal algorithm can not be handled by any existing approaches. For each case study, we verify the validity of some optimizations (parallelizing a traversal and/or fusing multiple traversals) using the MSO encoding approach set forth above.
Our framework leverages \Mona~\cite{mona}, a state-of-the-art WS2S (weak MSO with two successors) logic solver as our back-end constraint solver. All experiments were run on a server with a 40-core, 2.2GHz CPU and 128GB memory running Fedora 26. Remember our MSO encodings of data-race-freeness and equivalence are sound but not complete, the negative answers could be spurious. To this end, whenever \Mona returned a counterexample, we manually investigated if it corresponds to a real evidence of violation.

\vspace{-0.2cm}
\subsubsection*{Mutually Recursive Size-Counting.}
%

\begin{figure}[t!b]

\begin{subfigure}[b]{0.5\columnwidth}
\begin{lstlisting}[basicstyle=\sffamily \scriptsize]
  Fused(n)
    if (n == nil) return 0
    else
    	(ls, lv) = Fused(n.l)
    	(rs, rv) = Fused(n.r)
    	return (ls + rs + 1, lv + rv)
\end{lstlisting}
%
\caption{A valid fusion}
\label{fig:fusiblesize}
\end{subfigure}
\begin{subfigure}[b]{0.5\columnwidth}
\begin{lstlisting}[basicstyle=\sffamily \scriptsize]
  Fused(n)
    if (n == nil) return 0
    else
    	(ret1, ret2) = (ls + rs + 1, lv + rv)
    	(ls, lv) = Fused(n.l)
    	(rs, rv) = Fused(n.r)
    	return (ret1, ret2)
\end{lstlisting}
%
\caption{An invalid fusion}
\label{fig:infusiblesize}
\end{subfigure}
\caption{Fusing the two mutually recursive traversals in the running example}
\label{fig:fusedsize}
\end{figure}

This is our running example presented in Figure~\ref{fig:number}. We verified that the mutually recursive traversals \code{Odd} and \code{Even} can be fused to a single traversal shown in Figure~\ref{fig:fusiblesize} (solved by {\sc Mona} in 0.14s). This simple verification task, to our knowledge, is already beyond the capability of existing approaches. 
We also designed an invalid fused traversal (shown in Figure~\ref{fig:infusiblesize}) and encode the fusibility to MSO. {\sc Mona} returned a counterexample in 0.14s that illustrates how the data dependence is violated. Basically, the read-after-write dependence between a child and its parent in traversal \code{Even} is violated after the fusion. We manually verified that the counterexample is a true positive.

We also checked the data-race-freeness of the original program. The two parallel traversals \code{Odd(n)} and \code{Even(n)} in the \code{main} function  are independent because in every layer of the tree there is exactly one \code{Odd} call and one \code{Even} call and they belong to different traversal on each layer of the tree. The data-race-freeness was checked in 0.02s.

\vspace{-0.2cm}
\subsubsection*{Tree-Mutation.}

\begin{figure}[t!b]
\begin{subfigure}[c]{0.62\columnwidth}
\begin{subfigure}[c]{0.4\columnwidth}
\begin{lstlisting}[basicstyle=\sffamily \scriptsize]
  Swap(n)
    if (n == nil) return
    else
    	Swap(n.l)
    	Swap(n.r)
    	tmp = n.l
    	n.l = n.r
    	n.r = tmp
\end{lstlisting}
\end{subfigure}
\begin{subfigure}[c]{0.4\columnwidth}
\begin{lstlisting}[basicstyle=\sffamily \scriptsize]
  IncrmLeft(n)
    if (n == nil) return
    else
    	IncrmLeft(n.l)
    	IncrmLeft(n.r)
    	if (n == nil) n.v = 1
    	else n.v = n.l.v + 1
  Main(n)
    Swap(n)
    IncrmLeft(n)
\end{lstlisting}
\end{subfigure}
\caption{Before fusion}
\label{fig:unfusedswap}
\end{subfigure}
\hfill
\begin{subfigure}[c]{0.37\columnwidth}
\begin{lstlisting}[basicstyle=\sffamily \scriptsize]
  Fused(n)
    if (n == nil) return
    else
    	Fused(n.l)
    	Fused(n.r)
    	tmp = n.l
    	n.l = n.r
    	n.r = tmp
    	if (n == nil) n.v = 1
    	else n.v = n.l.v + 1
\end{lstlisting}
\caption{After fusion}
\label{fig:fusedswap}
\end{subfigure}
\caption{Fusing tree mutation traversal}
\label{fig:swap}
\end{figure}

%
%

We checked the fusion of two tree-mutating traversals. Figure~\ref{fig:unfusedswap} shows the two original traversals: \code{Swap} is a tree-mutating traversal that recursively swaps the sibling nodes of a binary tree; \code{IncrmLeft} updates the local field \code{n.v} depending on the value stored in its left child. Figure~\ref{fig:fusedswap} shows the fused traversal. 

Notice that \code{Swap} mutates the tree topology, which is disallowed in \name{}.
However, as mentioned in Section~\ref{sec:syntax}, mutation operations can be simulated using several mutable local fields. 
For example, for the statement \code{n.l = n.r}, we introduced two local boolean fields, \code{n.ll} for ``\code{n.l} is unchanged'' and \code{n.lr} for ``\code{n.l} is pointing to the original right child of \code{n}''. Before any manipulation, the two fields are initialized as \code{n.ll = true; n.lr = false}; and the statement \code{n.l = n.r} can be replaced with \code{n.lr = true; n.ll = false}. Then any other statement reading \code{n.l} will be converted to a conditional statement. For example, a call \code{f(n.l)} is converted to \code{if (n.ll) f(n.l) else if (n.lr) f(n.r)}.

After the conversion, we also rewrote the original and fused programs to eliminate some conditional branches using information extracted by a simple program analysis. For example, after swapping the siblings of \code{n}, \code{n.lr} is currently true at the program point, then \code{if (n.ll) IncrmLeft(n.l) else if (n.lr) IncrmLeft(n.r)} can be simplified as \code{IncrmLeft(n.r)}.
After these preprocessing steps, we obtained standard \name{} programs and the fusibility was checked by {\sc Mona} in 0.12s.

\vspace{-0.2cm}
\subsubsection*{CSS Minification.}

\begin{figure}[t!b]
\begin{subfigure}[c]{0.54\columnwidth}
\begin{lstlisting}[basicstyle=\sffamily \scriptsize]
  ConvertValues(n)
    if (n == nil) return 0
    else
    	for each child p: ConvertValues(n.p)
    	if (n.type == "word" || n.type == "func")
				n.value = TransValue(n.value)
  MinifyFont(n)
    if (n == nil) return 0
    else
    	for each child p: MinifyFont(n.p)
    	if (n.prop == "font-weight")
				n.value = MinifyWeight(n.value)
\end{lstlisting}
\end{subfigure}
\begin{subfigure}[c]{0.46\columnwidth}
\begin{lstlisting}[basicstyle=\sffamily \scriptsize]
  ReduceInit(n)
    if (n == nil) return 0
    else
    	for each child p: ReduceInit(n.p)
    	if (length(n.value) < initialLength)
				n.value = ReduceInitial(n.value)
  Main(n)
    ConvertValues(n)
    MinifyFont(n)
    ReduceInit(n)
\end{lstlisting}
\end{subfigure}
\caption{CSS minification traversals}
\label{fig:css}
\end{figure}



Cascading Style Sheets (CSS) is a widely-used style sheet language for web pages. In order to lessen the page loading time, many minification techniques are adapted to reduce the size of CSS document so that the time spent on delivering CSS document can be reduced~\cite{cssmin,cssnano,minify,csso,clean-css}. When minifying the CSS file, the Abstract Syntax Tree (AST) of the CSS code is traversed several times to perform different kinds of minifications, such as shortening identifiers, reducing whitespaces, etc. In the case that the same AST is traversed multiple times, fusing the traversals together would be desirable to enhance the performance of minification process.

Hence, we consider checking the fusibility of three CSS minification traversals shown in Figure~\ref{fig:css}. These traversals are similar to the ones presented in \cite{cssnano}. Traversal \code{ConvertValues} converts values to use different units when conversion result in smaller CSS size. For instance, \code{100ms} will be represented as \code{.1s}. Traversal \code{MinifyFont} will try to minimize the font weight in the code. For example, \code{font-weight: normal} will be rewritten to \code{font-weight: 400}. Traversal \code{ReduceInit} reduces the CSS size by converting the keyword \code{initial} to corresponding value when keyword \code{initial} is longer than the property value. For example, \code{min-width: initial} will be converted to \code{min-width: 0}.

Notice that these programs involve conditions on string which are not supported by \name{}. Nonetheless, since the traversals in Figure~\ref{fig:css} only manipulate the local fields of the AST, these conditions can be replaced by some simple arithmetic conditions. 
Moreover, as the ASTs of CSS programs are typically not binary trees and cannot be handled by {\sc Mona} directly, we also converted the ASTs to left-child right-sibling binary trees and then simplify the traversals to match \name{} syntax. The fusibility of the three minification traversals were checked in 6.88s.

\subsubsection*{Cycletree Routing.}

\begin{figure}[t!b]
\begin{subfigure}[c]{0.33\columnwidth}
\begin{lstlisting}[basicstyle=\sffamily \scriptsize]
  RootMode(n, number)
    if (n == nil) return
    else
    	n.num = number
    	number = number+1
    	PreMode(n.l, number)
    	PostMode(n.r, number)
  PreMode(n, number)
    if (n == nil) return
    else
    	n.num = number
    	number = number+1
    	PreMode(n.l, number)
    	InMode(n.r, number)
\end{lstlisting}
\end{subfigure}
\begin{subfigure}[c]{0.33\columnwidth}
\begin{lstlisting}[basicstyle=\sffamily \scriptsize]
  InMode(n, number)
    if (n == nil) return
    else
    	PostMode(n.l, number)
    	n.num = number
    	number = number+1
    	PreMode(n.r, number)
  PostMode(n, number)
    if (n == nil) return
    else
    	InMode(n.l, number)
    	PostMode(n.r, number)
    	n.num = number
    	number = number+1
\end{lstlisting}
\end{subfigure}
\begin{subfigure}[c]{0.32\columnwidth}
\begin{lstlisting}[basicstyle=\sffamily \scriptsize]
  ComputeRouting(n)
    if (n == nil) return
    else
    	ComputeRouting(n.l)
    	ComputeRouting(n.r)
    	n.lmin = n.l.min
    	n.rmin = n.r.min
    	n.lmax = n.l.max
    	n.rmax = n.r.max
    	n.max = MAX(n.lmax, n.rmax, n.num)
    	n.min = MIN(n.lmin, n.rmin, n.num)
  Main(n)
    RootMode(n, 0)
    ComputeRouting(n)
\end{lstlisting}
\end{subfigure}
\caption{Ordered cycletree construction and routing data computation}
\label{fig:cycletree}
\end{figure}

Our last and most challenging case study is about Cycletrees~\cite{cycletree}, a special class of binary trees with additional set of edges. These additional edges serve the purpose of constructing a Hamiltonian cycle. 
Hence, cycletrees are especially useful when it comes to different communication patterns in parallel and distributed computation. For instance, a broadcast can be efficiently processed by the tree structure while the cycle order is suitable for point-to-point communication. Cycletrees are proven to be an efficient network topology in terms of degree and number of communication links \cite{cycletree3,cycletree,cycletree2}.

Figure~\ref{fig:cycletree} shows the code snippet of two traversals over a cycletree. \code{RootMode} is a mutually recursive traversal that construct the cyclic order on a binary tree to transform the binary tree to a cycletree. \code{n.num} stores the order of current node {n} in the cyclic order of cycletree. \code{ComputeRouting} computes the router data of each node. The router data \code{n.lmin, n.rmin, n.lmax, n.rmax} are essential for an efficient cycletree routing algorithm that presented in~\cite{cycletree}. In the event of cyclic order traversal and routing had to be performed repeatedly---in case of link failures---it would be useful to think about ways we can optimize these procedures by fusion or parallelization.

We first consider checking the fusibility of these two traversals \code{RootMode} and \code{ComputeRouting}. 
We omit the fused traversal in the interest of space.
The total time spent to verify the fusibility of these two traversals was 490.55s.

We then considered whether the two traversals can run in parallel. This time \Mona spent 0.95s and returned a counterexample which allows us to discover a data race. Essentially, the counterexample illustrates that the read-after-write dependence over \code{n.num} between \code{PostMode} and \code{ComputeRouting} may be violated by parallelization. We manually check and verify that the counterexample is indeed a true positive.

\section{Related Work}
\label{sec:related}


There has been much prior work on program dependence analysis for tree data structures. Using shape analyses~\cite{shape82}, Ghiya~\etal~\cite{Ghiya1998} detect function calls that access disjoint subtrees for parallel computation in programs with recursive data structures. Rugina and Rinard~\cite{Rugina2000} extract symbolic lower and upper bounds for the regions of memory that a program accesses. Instead of providing a framework that describe dependences in programs, these work emphasizes only on detecting the data races and the potential of parallel computing so that is not able to handle fusion or other transformations.

Amiranoff~\etal~\cite{Amiranoff} propose instance-wise analysis to perform dependence analysis for recursive programs involving trees. This framework represents each dynamic instance of a statement by an execution trace, and then abstract the execution trace to a finitely-presented control word. Nonetheless, the framework does not support applications other than parallelization and they can not handle programs with tree mutation. Weijiang~\etal~\cite{Weijiang2015} also present a tree dependence analysis framework that reason legality of point blocking, traversal slicing and parallelization with the assumption that all traversals are identical preorder traversals. Their framework allows restricted tree mutations including nullifying or creating a subtree but the traversals that they consider are also single node traversals like \name{}. Deforestation~\cite{Wadler1990,Gill1993,MartiNez2013,Rompf2013,DAntoni2014} is a technique widely applied to fusion, but it either does not support fusion over arbitrary tree traversals, or does not handle reasoning about imperative programs.

The last decade has seen significant efforts on reasoning transformations over recursive tree traversals. Meyerovich~\etal~\cite{Meyerovich2010,Meyerovich2013} focus on fusing tree traversals over ASTs of CSS files. They specify tree traversals as attribute grammars and present a synthesizer that automatically fuse and parallelize the attribute grammars. Their framework only support traversals that can be written as attribute grammars, basically layout traversals. Rajbhandari~\etal~\cite{rajbhandari2016sc} provide a domain specific fusion compiler that fuse traversals of $k$-d trees in computational simulations. Both the frameworks are ad hoc, designed to serve specific applications. The tree traversals they can handle are less general than \name{}.

Most recently, TreeFuser~\cite{Sakka2017} presented by Sakka~\etal is an automatic framework that fuses tree traversals written in a general language. TreeFuser supports code motion and partial fusion, i.e., part of a traversal (left subtree or right subtree) can be fused together when possible, even if the traversals can not be fully fused. Their approach can not handle transformations other than fusion. In other words, parallelization of traversals is beyond the scope of TreeFuser. Besides, TreeFuser also suffers from the restrictions that \name{} have, i.e. no tree mutation and single node traversal. PolyRec~\cite{polyrec} is a framework that can handle schedule transformations for nested recursive programs only. PolyRec targets a limited class of tree traversals, called perfectly nested recursive programs, hence the framework is not able to handle arbitrary recursive tree traversals. Also PolyRec does not handle dependence analysis and suffers from the restriction that no tree mutation is allowed. The transformations that they handle are interchange, inlining and code motion rather than fusion and parallelization. None of the dependence analysis in the frameworks above is expressive enough to handle mutual recursion.

\bibliographystyle{splncs04}
\bibliography{refs,refs2}

\begin{thebibliography}{10}
\providecommand{\url}[1]{\texttt{#1}}
\providecommand{\urlprefix}{URL }
\providecommand{\doi}[1]{https://doi.org/#1}

\bibitem{Amiranoff}
Amiranoff, P., Cohen, A., Feautrier, P.: Beyond iteration vectors: Instancewise
  relational abstract domains. In: Yi, K. (ed.) Static Analysis. pp. 161--180.
  Springer Berlin Heidelberg, Berlin, Heidelberg (2006)

\bibitem{cssmin}
Bleuzen, J.: cssmin, \url{https://www.npmjs.com/package/cssmin}

\bibitem{cssnano}
Briggs, B., Contributers: cssnano, \url{https://cssnano.co/}

\bibitem{minify}
Clay, S., Contributers: minify, \url{https://github.com/mrclay/minify}

\bibitem{DAntoni2014}
D'Antoni, L., Veanes, M., Livshits, B., Molnar, D.: Fast: A transducer-based
  language for tree manipulation. SIGPLAN Not.  \textbf{49}(6),  384--394 (jun
  2014)

\bibitem{csso}
Dvornov, R., Contributers: csso, \url{https://github.com/css/csso}

\bibitem{mona}
Elgaard, J., Klarlund, N., M{\o}ller, A.: {MONA} 1.x: new techniques for {WS1S}
  and {WS2S}. In: Proc. 10th International Conference on Computer-Aided
  Verification, CAV~'98. LNCS, vol.~1427, pp. 516--520. Springer-Verlag
  (June/July 1998)

\bibitem{Ghiya1998}
Ghiya, R., Hendren, L.J., Zhu, Y.: Detecting parallelism in c programs with
  recursive darta structures. In: Proceedings of the 7th International
  Conference on Compiler Construction. pp. 159--173. CC '98, Springer-Verlag,
  London, UK, UK (1998), \url{http://dl.acm.org/citation.cfm?id=647474.727598}

\bibitem{Gill1993}
Gill, A., Launchbury, J., Peyton~Jones, S.L.: A short cut to deforestation. In:
  Proceedings of the Conference on Functional Programming Languages and
  Computer Architecture. pp. 223--232. FPCA '93, ACM, New York, NY, USA (1993).
  \doi{10.1145/165180.165214}, \url{http://doi.acm.org/10.1145/165180.165214}

\bibitem{jo11oopsla}
Jo, Y., Kulkarni, M.: Enhancing locality for recursive traversals of recursive
  structures. In: Proceedings of the 2011 ACM international conference on
  Object oriented programming systems languages and applications. pp. 463--482.
  OOPSLA '11, ACM, New York, NY, USA (2011).
  \doi{http://doi.acm.org/10.1145/2048066.2048104},
  \url{http://doi.acm.org/10.1145/2048066.2048104}

\bibitem{jo12oopsla}
Jo, Y., Kulkarni, M.: Automatically enhancing locality for tree traversals with
  traversal splicing. In: Proceedings of the 2012 ACM international conference
  on Object oriented programming systems languages and applications. OOPSLA
  '12, ACM, New York, NY, USA (2012)

\bibitem{shape82}
Jones, N.D., Muchnick, S.S.: A flexible approach to interprocedural data flow
  analysis and programs with recursive data structures. In: Proceedings of the
  9th ACM SIGPLAN-SIGACT Symposium on Principles of Programming Languages. pp.
  66--74. POPL '82, ACM, New York, NY, USA (1982). \doi{10.1145/582153.582161},
  \url{http://doi.acm.org/10.1145/582153.582161}

\bibitem{MartiNez2013}
Mart\'{\i}Nez, M., Pardo, A.: A shortcut fusion approach to accumulations. Sci.
  Comput. Program.  \textbf{78}(8),  1121--1136 (Aug 2013).
  \doi{10.1016/j.scico.2012.09.002},
  \url{http://dx.doi.org/10.1016/j.scico.2012.09.002}

\bibitem{Meyerovich2010}
Meyerovich, L., Bodik, R.: Fast and parallel webpage layout. In: Proceedings of
  the 19th international conference on world wide web. pp. 711--720. WWW '10,
  ACM (2010)

\bibitem{Meyerovich2013}
Meyerovich, L.A., Torok, M.E., Atkinson, E., Bodik, R.: Parallel schedule
  synthesis for attribute grammars. In: Proceedings of the 18th ACM SIGPLAN
  Symposium on Principles and Practice of Parallel Programming. pp. 187--196.
  PPoPP '13, ACM, New York, NY, USA (2013). \doi{10.1145/2442516.2442535},
  \url{http://doi.acm.org/10.1145/2442516.2442535}

\bibitem{Minsky1967}
Minsky, M.L.: Computation: Finite and Infinite Machines. Prentice-Hall, Inc.,
  Upper Saddle River, NJ, USA (1967)

\bibitem{clean-css}
Pawlowicz, J.: clean-css, \url{https://github.com/jakubpawlowicz/clean-css}

\bibitem{petrashko17}
Petrashko, D., Lhot\'{a}k, O., Odersky, M.: Miniphases: Compilation using
  modular and efficient tree transformations. In: Proceedings of the 38th ACM
  SIGPLAN Conference on Programming Language Design and Implementation. pp.
  201--216. PLDI 2017, ACM, New York, NY, USA (2017).
  \doi{10.1145/3062341.3062346},
  \url{http://doi.acm.org/10.1145/3062341.3062346}

\bibitem{rajbhandari2016sc}
Rajbhandari, S., Kim, J., Krishnamoorthy, S., Pouchet, L.N., Rastello, F.,
  Harrison, R.J., Sadayappan, P.: A domain-specific compiler for a parallel
  multiresolution adaptive numerical simulation environment. In: Proceedings of
  the International Conference for High Performance Computing, Networking,
  Storage and Analysis. pp. 40:1--40:12. SC '16, IEEE Press, Piscataway, NJ,
  USA (2016), \url{http://dl.acm.org/citation.cfm?id=3014904.3014958}

\bibitem{rajbhandari2016fusing}
Rajbhandari, S., Kim, J., Krishnamoorthy, S., Pouchet, L.N., Rastello, F.,
  Harrison, R.J., Sadayappan, P.: On fusing recursive traversals of kd trees.
  In: Proceedings of the 25th International Conference on Compiler
  Construction. pp. 152--162. ACM (2016)

\bibitem{Rompf2013}
Rompf, T., Sujeeth, A.K., Amin, N., Brown, K.J., Jovanovic, V., Lee, H.,
  Jonnalagedda, M., Olukotun, K., Odersky, M.: Optimizing data structures in
  high-level programs: New directions for extensible compilers based on
  staging. In: Proceedings of the 40th Annual ACM SIGPLAN-SIGACT Symposium on
  Principles of Programming Languages. pp. 497--510. POPL '13, ACM, New York,
  NY, USA (2013). \doi{10.1145/2429069.2429128},
  \url{http://doi.acm.org/10.1145/2429069.2429128}

\bibitem{Rugina2000}
Rugina, R., Rinard, M.: Symbolic bounds analysis of pointers, array indices,
  and accessed memory regions. In: Proceedings of the ACM SIGPLAN 2000
  Conference on Programming Language Design and Implementation. pp. 182--195.
  PLDI '00, ACM, New York, NY, USA (2000). \doi{10.1145/349299.349325},
  \url{http://doi.acm.org/10.1145/349299.349325}

\bibitem{Sakka2017}
Sakka, L., Sundararajah, K., Kulkarni, M.: Treefuser: A framework for analyzing
  and fusing general recursive tree traversals. Proc. ACM Program. Lang.
  \textbf{1}(OOPSLA),  76:1--76:30 (Oct 2017). \doi{10.1145/3133900},
  \url{http://doi.acm.org/10.1145/3133900}

\bibitem{polyrec}
Sundararajah, K., Kulkarni, M.: Composable, sound transformations of nested
  recursion and loops. In: Proceedings of the 40th ACM SIGPLAN Conference on
  Programming Language Design and Implementation. pp. 902--917. PLDI 2019, ACM,
  New York, NY, USA (2019). \doi{10.1145/3314221.3314592},
  \url{http://doi.acm.org/10.1145/3314221.3314592}

\bibitem{cycletree3}
Veanes, M., Barklund, J.: Construction of natural cycletrees. Inf. Process.
  Lett.  \textbf{60}(6),  313--318 (1996). \doi{10.1016/S0020-0190(96)00179-2},
  \url{https://doi.org/10.1016/S0020-0190(96)00179-2}

\bibitem{cycletree}
Veanes, M., Barklund, J.: Natural cycletrees: Flexible interconnection graphs.
  J. Parallel Distrib. Comput.  \textbf{33},  44--54 (02 1996).
  \doi{10.1006/jpdc.1996.0023}

\bibitem{cycletree2}
Veanes, M., Barklund, J.: On the number of edges in cycletrees. Inf. Process.
  Lett.  \textbf{57}(4),  225--229 (1996). \doi{10.1016/0020-0190(95)00183-2},
  \url{https://doi.org/10.1016/0020-0190(95)00183-2}

\bibitem{Wadler1990}
Wadler, P.: Deforestation: transforming programs to eliminate trees.
  Theoretical Computer Science  \textbf{73}(2),  231 -- 248 (1990).
  \doi{https://doi.org/10.1016/0304-3975(90)90147-A},
  \url{http://www.sciencedirect.com/science/article/pii/030439759090147A}

\bibitem{extended}
Wang, Y., Liu, J., Zhang, D., Qiu, X.: Reasoning about recursive tree
  traversals (2019), \url{https://arxiv.org/abs/1910.09521}

\bibitem{Weijiang2015}
Weijiang, Y., Balakrishna, S., Liu, J., Kulkarni, M.: Tree dependence analysis.
  In: Proceedings of the 36th ACM SIGPLAN Conference on Programming Language
  Design and Implementation. pp. 314--325. PLDI '15, ACM, New York, NY, USA
  (2015). \doi{10.1145/2737924.2737972},
  \url{http://doi.acm.org/10.1145/2737924.2737972}

\end{thebibliography}

\newpage
\appendix

\section{Proof of Theorem~1}
\label{sec:proof}
\begin{proof}
We prove the undecidability through a reduction from the halting problem of 2-counter machines~\cite{Minsky1967}. We can build a \name{} program to simulate the execution of a 2-counter machine. Given a 2-counter machine $M$, every line of non-halt instruction $c$ in $M$ can be converted to a function in a \name{} program. The function is of the form $f_c({\sf n, v_1, v_2})$: \code{n} is a \code{Loc} parameter and ${\sf v_1, v_2}$ are \code{Int} parameters. It treats ${\sf v_1, v_2}$ as the current values of the two counters, updates the two counter values to ${\sf u_1, u_2}$ by simulating the execution of $c$, then recursively calls $f_{c'}({\sf n.l})$ if $c'$ the next instruction. for the halt instruction, a special function $f_{\textit{halt}}$ will pass up the signal by recursive calls, and finally run a special line of code ${\sf s}$ on the root. Then $M$ halts if and only if the iteration $({\sf s}, \textit{root})$ occurs.
\end{proof}

\section{Code Blocks}
\label{sec:block}

We introduce some necessary notations for blocks, of which the meaning is determined by the syntactic structure of the program. Figure~\ref{fig:sets} lists common sets of functions, blocks, parameters and nodes that will be frequently used in this paper.
We then define the possible relations between blocks. Figure~\ref{fig:relations} shows all the possible relations.
Every function's body can be represented as a syntax tree whose leaves are statement blocks and non-leaf nodes are sequentials, conditionals or parallels. Then the relation between two statement blocks is determined by their positions in the syntax tree. In particular, when two blocks ${\sf s} \sim {\sf t}$ belong to the same function ${\sf f}$, there are three possible relations, determined by the least common ancestor (LCA) node of ${\sf s}$ and ${\sf t}$ that is a sequential, conditional or parallel. 

\begin{figure}[t!]
\centering
\begin{tabular}{|c|l|}
\hline
 $\allfuncs$ & the set of all functions \\ \hline
 $\allparams$ & the set of all ${\sf Int}$ function parameters \\ \hline
 $\allblocks$ & the set of all blocks \\ \hline
 $\allcalls$ & the set of all blocks for function calls \\ \hline
 $\allnoncalls$ & the set of all blocks for straight-line non-call assignments \\ \hline
 $\blocks({\sf f})$ & the set of all blocks belonging to a function ${\sf f}$ \\ \hline
 $\params({\sf f})$ & the set of ${\sf Int}$ parameters for ${\sf f}$ \\ \hline
 $~~~{\sf Nodes}(T)~~~$ & the set of all nodes in the tree $T$ \\ \hline
 $\code{Path}(\code{t})$ & the path to \code{t} from the entry point of the function that \code{t} belongs to \\ \hline
\end{tabular}
\caption{Commonly Used Notations}
\label{fig:sets}
\end{figure}


\begin{example}
In our running example (Figure~\ref{fig:number}), there are 11 blocks. We number the blocks with \code{s0} through \code{s10}, as shown in the comment following each block. There are six call blocks: $\allcalls = \{ {\sf s1, s2, s5, s6, s8, s9} \}$; and five non-call blocks: $\allnoncalls = \{ {\sf s0, s3, s4, s7, s10} \}$. 
Take ${\sf s6}$ for example, $\code{Path}(\code{s6})$ is just the path from the beginning of function ${\sf Even}$ (which ${\sf s6}$ belongs to) to ${\sf s6}$, i.e., from $\neg~{\sf c1}$ to ${\sf s5}$ then ${\sf s6}$.
The $\sim$ relation holds between any two blocks from the same group: ${\sf s0}$ through ${\sf s3}$, ${\sf s4}$ through ${\sf s7}$, or ${\sf s8}$ through ${\sf s10}$.  ${\sf s2} \triangleleft {\sf s7}$ because ${\sf s2}$ calls ${\sf Even}$ and ${\sf s7} \in \blocks({\sf Even})$; ${\sf s5} \prec {\sf s7}$ because ${\sf s5}$ precedes ${\sf s7}$; ${\sf s0} \uparrow {\sf s1}$ because ${\sf s0}$ belongs to the if-branch and ${\sf s1}$ belongs to the else-branch; ${\sf s8} \parallel {\sf s9}$ because they are running in parallel.
\end{example}

\begin{figure}[t!]
\centering
\begin{tabular}{|c|l|}
\hline
 $\code{LCA(s, t)}$ & The least common ancestor (LCA) of blocks ${\sf s}$ and ${\sf t}$ in the syntax tree. \\ \hline
 $\code{s} \triangleleft \code{t}$ & $\code{s}$ is a function call to ${\sf f}$ and $\code{t} \in \blocks({\sf f})$. \\ \hline
 $\code{s} \sim \code{t}$ & $\code{s}$ and $\code{t}$ are from the same function definition, i.e., ${\sf s}, {\sf t} \in \blocks({\sf f})$ \\
 & for some function ${\sf f}$. \\ \hline
 ${\sf s} \prec {\sf t}$ & $\code{LCA(s, t)}$ is a sequential, i.e., ${\sf s}$ precedes ${\sf t}$. \\ \hline
 ${\sf s} \uparrow {\sf t}$ & $\code{LCA(s, t)}$ is a conditional, i.e., there is a conditional \code{if (...) then A else B} \\
 & such that ${\sf s}$ and ${\sf t}$ belong to \code{A} and \code{B}, respectively. \\ \hline
 $~~~~~{\sf s} \parallel {\sf t}~~~~~$ & $\code{LCA(s, t)}$ is a parallel, i.e., ${\sf s}$ and ${\sf t}$ can be executed in arbitrary order. \\ \hline
\end{tabular}
\caption{Relations Between Blocks}
\label{fig:relations}
\end{figure}

\begin{lemma}
For any two statement blocks ${\sf s}$ and ${\sf t}$, ${\sf s} \sim {\sf t}$ if and only if exactly one of the following relations holds: ${\sf s} \prec {\sf t}$, ${\sf s} \uparrow {\sf t}$ and ${\sf s} \parallel {\sf t}$.
\end{lemma}

\subsubsection*{Read\&Write analysis.}
\label{sec:readwrite}
In our framework, data dependences are represented and analyzed at the block level. We perform a static analysis over the program to extract the sets of local fields and variables being accessed in each non-call block. Intuitively, we use several read sets and write sets to represent local fields and global variables being read or written, respectively, in each statement block.

For every non-call block \code{s}, we build the read set $R_{\sf s}$ by adding all data fields and local variables occurred in an if-condition or on the RHS of an assignment. The data fields can be from the current node (such as \code{n.v}) or a neighbor node (such as \code{n.l.v}).
The write set $W_{\sf s}$ can be built similarly: all data fields and local variables occurred on the LHS of an assignment are added. 

\section{Formulating Reachability}
\label{sec:wp}

Note that the speculative execution of a function is completely deterministic as all initial parameters and return values from function calls are determined by $M$. More specifically, for every code snippet \code{l} without branching and every logical constraint $\varphi$ that should be satisfied after running \code{l}, we  can compute the weakest precondition $\code{wp}(\code{l}, \varphi, M)$ that must be satisfied before running \code{l}. The definition of \code{wp} is shown in Figure~\ref{fig:wp}.


\begin{figure}[!tb]
\begin{displaymath}
\begin{array}{lcl}
\code{wp}(n.f = \AExpr,~ \varphi,~ M) & = & \varphi[\AExpr / n.f] \\
\code{wp}(v = \AExpr,~ \varphi,~ M) & = & \varphi[\AExpr / v] \\
\code{wp}(\bar{v} = t(\dots),~ \varphi,~ M) & = & \varphi[M(\code{s}) / v] \quad \textrm{where } \code{s} \textrm{ is the id of the current statement} \\
\code{wp}(\code{l ; l'},~ \varphi,~ M) & = & \code{wp}(\code{l}, \code{wp}(l', \varphi)) \\
\end{array}
\end{displaymath}
\caption{Weakest Precondition}
\label{fig:wp}
\end{figure}

Now if ${\sf s}$ is a call to function $g$, we can determine if the speculative execution of $g$ with respect to $M$ hits block ${\sf t}$.
The path from the entry point of $g$ to \code{t} will be a straight-line sequence of statements of the form 
\begin{equation*}\label{eq:code}
\code{l}_1; \code{assume}(\code{c}_1); \dots; \code{assume}(\code{c}_{n-1}); \code{l}_n; \code{t}
\end{equation*}
where every branch condition is converted to a corresponding $\code{assume}(\code{c}_i)$.
Then we can compute the path condition for \code{t} by computing the weakest precondition for every condition $\code{c}_i$ on the path:
\begin{equation*}\label{eq:wp}
\textit{WP}(\code{c}_i, M) \equiv \code{wp}(\code{l}_1;\dots;\code{l}_i,~ \code{c}_i,~ M)[{M(\bar{p})} / \bar{p}]
\end{equation*}
where $\bar{p}$ is the sequence of arguments for $g$.

Moreover, when \code{t} is another call block, we also need to make sure the initial parameters in $N$ match the speculative execution of the above code sequence w.r.t. $M$. We denote this condition as $\textit{Match}_{\code{s,t}}(u, v, M, N)$. 


\end{document}